%% file: WeightedBWT.tex
\pdfoutput=1
\documentclass[runningheads]{llncs}
\usepackage[dvipsnames]{xcolor}
\usepackage[figuresright]{rotating}
\usepackage{graphicx}
\usepackage{float}
\usepackage{pythontex}
\usepackage{pgfplots}
\usepackage{dsfont}
\usepackage{bbm}
\usepackage{amsmath}
\usepackage{amssymb}
\usepackage{rotating}
\usepackage{booktabs}
\usepackage{multirow}
\usepackage{tikz}
\usepackage{hyperref} 
\usepackage[dvipsnames]{xcolor}

\pgfplotsset{compat=1.14}

\definecolor{mygray}{gray}{0.8}
\begin{document}

\title{Weighted Burrows-Wheeler Compression} 

\titlerunning{Weighted Burrows-Wheeler Compression} 

\author{Aharon Fruchtman \inst{1}
\and
Yoav Gross \inst{1}
\and\\
Shmuel T.\ Klein\inst{2}
\orcidID{0000-0002-9478-3303} 
\and 
Dana Shapira \inst{1} \orcidID{0000-0002-2320-9064}}

\authorrunning{A.\ Fruchtman et al.}

\institute{Computer Science Department, Ariel University,  Israel 
\email{\{aralef,yodgimmel\}@gmail.com, shapird@g.ariel.ac.il}\\
\and
Computer Science Department, Bar Ilan University,  Israel\\
\email{tomi@cs.biu.ac.il}}

\maketitle


\begin{abstract}
    
A weight based dynamic compression method has  recently been proposed, which is especially suitable for the encoding of files with locally skewed distributions. 
Its main idea is to assign larger weights to closer to be encoded symbols by means of an increasing weight function, rather than considering each position in the text evenly.
A well known transformation that tends to convert input files into files with a more skewed distribution is the {\it Burrows-Wheeler Transform\/}. 
This paper employs the weighted approach on Burrows-Wheeler transformed files and provides empirical evidence of the efficiency of this combination.

\end{abstract}

\section{Introduction}

The {\it Burrows-Wheeler Transform\/} ({\sf BWT}) \cite{BWT} is the basis of the popular compression method {\sf bzip2}, yielding, on many types of possible input files, better compression than {\sf gzip} and other competitors. As a matter of fact, {\sf BWT} itself is not a compression method: its output is a permutation of its input, which has obviously the same size. The usefulness of the transformation is that it has a tendency to reorganize the data into what seems to be a more coherent form, grouping many, though not all, identical characters together. The output is therefore usually more compressible, by applying as simple methods as run-length coding and move-to-front.

The combination of different methods to be applied one after the other, an action known as {\it cascading\/}, is not new to data compression. It can, for example, be found in {\sf gzip}, which first parses the input using LZ77 \cite{LZ77}, and then applies Huffman coding \cite{Moffat19} to the parsed elements. It therefore seems natural to try to apply {\sf BWT} in a pre-processing stage, and then compress the transformed string by means of more sophisticated compression schemes. This strategy fails, however, for static or dynamic Huffman or arithmetic coding. The contribution of this paper is to give empirical evidence that cascading {\sf BWT} with the recently developed {\it weighted\/} compression schemes implies significant savings.

A family of dynamic compression algorithms, named {\it weighted coding}, has  recently been proposed, which is especially suitable for the encoding of files with locally skewed distributions. 
The main idea of the weighted approach is to assign larger weights to closer to be encoded symbols by means of an increasing weight function, rather than considering each position in the text evenly.

Traditional dynamic algorithms use the distribution of the symbols in the already processed portion of the input file as an estimate for the distribution of the elements still to come later in the text. However, the assumption that the past is a good approximation for the future, is not necessarily true.
A {\it Forward-looking\/} dynamic algorithm, using the true distribution of the remaining portion of the file, was suggested in \cite{KSSJ}. In this method the frequencies of the symbols in the entire file are prepended to the compressed file. In the encoding process, these frequencies are gradually updated to reflect the number of occurrences in the remaining part of the file by {\it decrementing\/} the frequency of the character that is currently being processed. A hybrid method, combining both traditional and forward-looking approaches, is proposed in \cite{FKS}: in this method the frequencies for all characters are not transmitted at the beginning of the file but rather progressively, each time a new character is encountered.

Pushing this approach even further, a family of {\it Forward weighted} coding schemes is proposed in \cite{arxiv}, in which the encoding is based on the distribution derived from index-dependent weights. That is, this weighted method assigns higher priorities to positions that are closer to the currently processed one in the encoding process, rather than treating all positions in the input file equally. 
The weights assigned to the positions are generated by a non increasing function $f$, and the weight for each symbol $\sigma$ is the sum of the values of $f$ over all the positions where $\sigma$ occurs in the portion of the input file that is still to be encoded.

Recently,  a specific variant named the {\it Backward Weighted\/} coding has been studied \cite{WeightedDCC21}, which suggests a heuristic based on a weighted distribution, calculated only over positions that have already been processed. The core advantage of such a heuristic approach is a negligible header, relatively to a costly header used in {\it Forward Looking} and all {\it Forward Weighted} variants. 
Empirical tests have shown that backward weighted techniques can improve beyond the lower bound given by the entropy for static encoding.

The paper is constructed as follows.
Section~2 reviews the notation and formulation needed for the weighted encodings using a running example. Section~3 proves some properties involving the combination of general dynamic coding techniques and {\sf BWT}.  Section~4 presents empirical outcomes supporting the compression efficiency of the proposed method even in practice.  

\section{Notation and Formulation}

For the completeness of the paper, we include the definitions given in \cite{arxiv}, which formalize entropy based compression methods.

Let $T=T[1,n]$ be an input file of size $n$ over an alphabet $\Sigma$ 
of size $m$. 
A weight $W$ is defined based on the following parameters,
\vspace{3mm}
\begin{itemize}
\item [-]
A non negative function $g$, $g:[1,n] \longrightarrow {\rm I\!R}^+$, which assigns a positive real number to integers, seen as an assignment of a weight to each position $i\in [1,n]$ within $T$;

\item [-]
A symbol of the alphabet, $\sigma\in \Sigma$;

\item [-]
An interval $[\ell, u]$, $1\le \ell\le u\le n$ for restricting the domain of the function $g$.
\end{itemize}
\vspace{3mm}
The value of $W(g,\sigma, \ell, u)$ is defined for each symbol $\sigma$, as the sum of the values of the function $g$ for all positions $j$ in the range $[\ell,u]$ at which $\sigma$ occurs. 
Formally,
$$ W(g,\sigma, \ell, u)\; = \sum_{\{\ell \leq j \le  u \ \mid \ T[j]=\sigma\}}{g(j)}.$$

As a special case of weighted coding, {\it Backward Weighted\/} considers all the positions that have already been processed, that is, the interval is of the form $[\ell, u]=[1,i-1]$, and
$$W(g,\sigma,1,i-1)\;=\sum_{\{1 \leq j \leq i-1 \ \mid \ T[j]=\sigma\}}{g(j)}.$$ 

Traditional encoding methods can be reformulated as special instances of $W$ for which $g=\mathds{1}\equiv g(i)=1$ for all $i$.
For example, 
{\sf static} compression refers to weights for which  $W(g,\sigma,\ell,u)=W(\mathds{1},\sigma,1,n)$ is constant for all indices.

\def\tstrut{\vrule height 9pt depth 3pt width 0pt}

\renewcommand{\arraystretch}{1}

\begin{sidewaystable}
\begin{center}
{\scriptsize
\begin{tabular}{@{\hspace{0mm}}l@{\hspace{-2mm}}c@{\hspace{2.5mm}}l@{\hspace{4mm}}l@{\hspace{2.5mm}}l@{\hspace{2.5mm}}l@{\hspace{2.5mm}}l@{\hspace{2.5mm}}l@{\hspace{2.5mm}}l@{\hspace{2.5mm}}l@{\hspace{2.5mm}}l@{\hspace{2.5mm}}l@{\hspace{2.5mm}}l@{\hspace{2.5mm}}l@{\hspace{2.5mm}}l@{\hspace{2.5mm}}l@{\hspace{2.5mm}}l@{\hspace{2.5mm}}l@{\hspace{2.5mm}}l@{\hspace{2.5mm}}l@{\hspace{2.5mm}}l@{\hspace{2.5mm}}l@{\hspace{0mm}}}
\toprule
&\multicolumn{1}{@{\hspace{0mm}}l}{$i$}&&\multicolumn{1}{@{\hspace{-0.5mm}}l}{\makebox[4mm]{\scriptsize 1}}& \multicolumn{1}{@{\hspace{-0.5mm}}l}{\makebox[4mm]{\scriptsize 2}}& 
\multicolumn{1}{@{\hspace{-0.5mm}}l}{\makebox[4mm]{\scriptsize 3}}&
\multicolumn{1}{@{\hspace{-0.5mm}}l}{\makebox[4mm]{\scriptsize 4}}&
\multicolumn{1}{@{\hspace{-0.5mm}}l}{\makebox[4mm]{\scriptsize $\cdots$}}& 
\multicolumn{1}{@{\hspace{-0.5mm}}l}{\makebox[4mm]{\scriptsize 13}}&
\multicolumn{1}{@{\hspace{-0.5mm}}l}{\makebox[4mm]{\scriptsize 14}}&
\multicolumn{1}{@{\hspace{-0.5mm}}l}{\makebox[4mm]{\scriptsize 15}}&
\multicolumn{1}{@{\hspace{-0.5mm}}l}{\makebox[4mm]{\scriptsize 16}}&
\multicolumn{1}{@{\hspace{-0.5mm}}l}{\makebox[4mm]{\scriptsize $\cdots$}}& 
\multicolumn{1}{@{\hspace{-0.5mm}}l}{\makebox[4mm]{\scriptsize 35}}& 
\multicolumn{1}{@{\hspace{-0.5mm}}l}{\makebox[4mm]{\scriptsize 36}}& 
\multicolumn{1}{@{\hspace{-0.5mm}}l}{\makebox[4mm]{\scriptsize 37}}& 
\multicolumn{1}{@{\hspace{-0.5mm}}l}{\makebox[4mm]{\scriptsize 38}}&
\multicolumn{1}{@{\hspace{-0.5mm}}l}{\makebox[4mm]{\scriptsize $\cdots$}}& 
\multicolumn{1}{@{\hspace{-0.5mm}}l}{\makebox[4mm]{\scriptsize 49}}& 
\multicolumn{1}{@{\hspace{-0.5mm}}l}{\makebox[4mm]{\scriptsize 50}}& 
\multicolumn{1}{@{\hspace{0mm}}l@{\hspace{0mm}}}{\makebox[4mm]{\sf avg}} \\[1mm]
& $T$ && {\tt a} & {\tt t} & {\tt a} & {\tt t} & $\cdots$ & {\tt a} & {\tt t} & {\tt c} & {\tt g} & $\cdots$ & {\tt c} & {\tt g} & {\tt a} & {\tt t} & $\cdots$ & {\tt a} & {\tt t} \\[1mm] \toprule
& & $W$ & \tstrut 14 & 14 & 14 & 14 & ... & 14 & 14 & 11 & 11 & $\cdots$ & 11 & 11 & 14 & 14 & $\cdots$ & 14 & 14 & 
\\ [1.5mm]
{\sf static}& & {\sf TotIndx} & \tstrut 50 & 50 & 50 & 50 & ... & 50 & 50 & 50 & 50 & $\cdots$ & 50 & 50 & 50 & 50 & $\cdots$ & 50 & 50 & {\sf 1.990} \\ [2.5mm]
& &{\sf IC} & 1.84 & 1.84 & 1.84 & 1.84 &$\cdots$ & 1.84 & 1.84 & 2.18 & 2.18 & $\cdots$ & 2.18 & 2.18 & 1.84 & 1.84 & $\cdots$ & 1.84 & 1.84 & {\sf (+0.291)} \\ [1.5mm] 
\cline{3-21}
& & {\sf Indx}$W$  & \tstrut 1 & 1 & 1 & 1 &$\cdots$ & 1 & 1 & 1 & 1 & $\cdots$ & 1 & 1 & 1 & 1 & $\cdots$ & 1 & 1 &  \\ [1.5mm]
& & $W$ & \tstrut 1 & 1 & 2 & 2 & $\cdots$ & 7 & 7 & 1 & 1 & $\cdots$ & 11 & 11 & 8 & 8 & $\cdots$ & 14 & 14 &  \\ [1.5mm]
{\sf b-adp}& & {\sf TotIndx} & \tstrut 4 & 5 & 6 & 7 & $\cdots$ & 16 & 17 & 18 & 19 & $\cdots$ & 38 & 39 & 40 & 41 & $\cdots$ & 52 & 53 & 
{\sf 2.111} \\ [2.5mm]
& &{\sf IC} & 2.00 & 2.32 & 1.58 & 1.81 & $\cdots$ & 1.19 & 1.28 & 4.17 & 4.25 & $\cdots$ & 1.79 & 1.83 & 2.32 & 2.36 & $\cdots$ & 1.89 & 1.92 \\ [1.5mm] \cline{3-21}
& & {\sf Indx}$W$  & \tstrut 1 & 1 & 1 & 1 & $\cdots$ & 1 & 1 & 1 & 1 & $\cdots$ & 1 & 1 & 1 & 1 & $\cdots$ & 1 & 1 &  \\ [1.5mm]
& & $W$ & \tstrut 14 & 14 & 13 & 13 & $\cdots$ & 8 & 8 & 11 & 11 & $\cdots$ & 1 & 1 & 7 & 7 & $\cdots$ & 1 & 1 &  \\ [1.5mm]
{\sf f-adp}& & {\sf TotIndx} & \tstrut 50 & 49 & 48 & 47 & $\cdots$ & 38 & 37 & 36 & 35 & $\cdots$ & 16 & 15 & 14 & 13 & $\cdots$ & 2 & 1 & 
{\sf 1.820} \\ [2.5mm]
& &{\sf IC} & 1.82 & 1.84 & 1.81 & 1.88 & $\cdots$ & 2.25 & 2.21 & 1.71	& 1.67 & $\cdots$ & 4.00 & 3.91 & 1.00 & 0.89 & $\cdots$ & 1.00 & 0.00 & {\sf (+0.291)} \\ [1.5mm] \cline{3-21}
& & {\sf Indx}$W$ &\tstrut  1 & 1 & 1 & 1 & $\cdots$ & 4 & 4 & 4 & 8 & $\cdots$ & 64 & 128 & 128 & 128 & $\cdots$ & 512 & 512 \\ [1.5mm]
\multirow{2}{*}{\raisebox{-1mm}{\sf b-2}} & & $W$ &\tstrut  1 & 1 & 2 & 2 & $\cdots$ & 12 & 13 & 1 & 1 & $\cdots$ & 261 & 281 & 16 & 17 & $\cdots$ & 1552 & 1809 &\multirow{2}{*}
{\raisebox{-1.5mm}{\sf 1.981}}\\ [1.5mm] 
& & {\sf TotIndx} & \tstrut 4 & 5 & 6 & 7 & $\cdots$ & 27 & 31 & 35 & 39 & $\cdots$ & 575 & 639 & 767 & 895 & $\cdots$ & 4095 & 4607 \\ [2.5mm]
& & {\sf IC} & 2.00 & 2.32 & 1.58 & 1.81 & $\cdots$ & 1.17 & 1.25 & 5.13 & 5.29 & $\cdots$ & 1.14 & 1.19 & 5.58 & 5.72 & $\cdots$ & 1.40 & 1.35 \\ [1.5mm]
\cline{3-21}
& & {\sf Indx}$W$ &\tstrut  1.00 & 1.15 & 1.32 & 1.52 & $\cdots$ & 5.28 & 6.06 & 6.96 & 8.00 & $\cdots$ & 111.43 & 128.00 & 147.03 & 168.90 & $\cdots$ & 776.05 & 891.44 \\ [1.5mm]
\multirow{2}{*}{\sf b-weight} & & $W$ &\tstrut  1.00 & 1.00 & 2.00 & 2.15 &  $\cdots$ & 14.39	& 16.38 & 1.00 & 1.00 & $\cdots$ & 327.96 & 376.58 & 19.67 & 22.44 & $\cdots$ & 1988.36 & 2283.88&\multirow{2}{*}
{\raisebox{-1.5mm}{\sf 1.989}}  \\ [1.5mm] 
& & {\sf TotIndx} & \tstrut 4.00 & 5.00 & 6.15 & 7.47 &
$\cdots$ & 32.77 & 38.05 & 44.11 & 51.08 & $\cdots$ & 746.65 & 858.08 & 986.08 & 1113.11 & $\cdots$ & 5216.21 & 5992.26 \\ [2.5mm]
& & {\sf IC} & 2.00 & 2.32 & 1.62 & 1.80 & $\cdots$ & 1.19 & 1.22	& 5.46 & 5.67 & $\cdots$ & 1.19 & 1.19 & 5.65 & 5.66 & $\cdots$ & 1.39 & 1.39 \\ [1.5mm]
\bottomrule
\end{tabular}
}
\vspace*{5mm}
\caption{Coding example for $T={\tt {(at)}^7{(cg)}^{11}{(at)}^7}$.}
\label{example}
\end{center}
\end{sidewaystable}


As a short illustration for the weighted approach, 
Table~1 brings a comparative chart for the encoding of a small example of 50 characters: 
$$T=x_1\cdots x_{50}={\tt (at)}^{7}{\tt (cg)}^{11}{\tt (at)}^{7},$$ 
shown in the second row of the table for several representative portions of $T$, just underneath the indices. 
The {\sf static} compression for $T$ considers the probabilities $\frac{14}{50}$ for {\tt a} and {\tt t} and $\frac{11}{50}$ for {\tt c} and {\tt g}. 
These probabilities can be calculated from the two first rows corresponding to the {\sf static} method  in Table~1, the first entitled  $W$ and representing the weight of the specific character,  and the second  entitled {\sf TotIndx} representing the cumulative  
weights of all the characters. The ratio of these weights can be considered as a probability $p_i$, and the
corresponding {\it Information content\/}, $-\log p_i$  for each position $i$, is shown in the line entitled {\sf IC}. For {\sf static}, the {\sf IC} values are 
1.84 for {\tt a} and {\tt t} and  2.18 for {\tt c} and {\tt g}. The last column of the table, headed  {\sf avg}, gives the weighted average of these {\sf IC} values, which is in fact the {\it entropy\/}. 

The classic adaptive coding, {\sf b-adp} \cite{Vitter87}, is a specific backward weight method in which $g(i)=1$ for all $i$, and the backward weights refer to all positions $i$ with $1 \le i < n$ by:
$$W(\mathds{1},\sigma,1,i-1)= \sum_{\{1\le j \le i-1 \ \mid \ T[j]=\sigma\}}{1},$$
which is simply the number of occurrences of the current character $\sigma$ up to that point, i.e.,  in  $T[1,i-1]$. The details appear in the second part of Table~1, headed {\sf b-adp}. The line entitled $W$ refers now, at position $i$, to specific weight of the character $\sigma=T[i]$ 
up to the given column $i$ of the table, that is, the sum of the index weights {\sf Indx}$W$ for those indices $j<i$ at which the character $\sigma$ occurs,  including the initial 1 values. For {\sf b-adp}, as well as for the following method {\sf f-adp}, the values of {\sf Indx}$W$ are just 1 for every $i$. The cumulative $W$ values for all the characters $\sigma\in\Sigma$ are given in the row entitled {\sf TotIndx}. 

The symmetric counterpart of {\sf b-adp} is the {\it forward looking\/} method, {\sf f-adp}, which, at each location $i$, considers the positions yet to come $[i,n]$ rather than those already processed $[1,i-1]$ as the {\it backward looking\/} {\sf b-adp}. That is, $$W(\mathds{1},\sigma,i,n)= \sum_{\{i\le j \le n \ \mid \ T[j]=\sigma\}}{1}.$$
For our running example, {\sf f-adp} initializes the weights of the characters {\tt a, t} to 14 and {\tt c, g} to 11. The counts are then gradually decremented reflecting the remaining number of occurrences for each $\sigma$ from position $i$ to the end of $T$.
The header for {\sf f-adp} describes the exact frequencies of the involved characters, and its size is 0.291 for our example, as for {\sf static}.

A simple adaptive weighted coding, denoted by {\sf b-2}, has been proposed in \cite{WeightedDCC21}. It is 
inspired by Nelson and Gailly \cite{Nelson}, who rescale the frequencies in order to cope with hardware constraints like the representation of the number of occurrences as 16-bit integers.   
{\sf b-2} divides all the frequencies at the end of every block of $k$ characters, for a given parameter $k$, regardless of computer hardware restrictions. That is, the occurrences of characters at the beginning of the input file contribute to $W$ less than those closer to the current position. Furthermore, all positions within the same block contribute equally to $W$, and  their contribution weight is twice as large as the weight assigned to the indices in the preceding block. Formally, for each pair of indices $i$ and $i+k$, the function $g$ for {\sf b-2}, denoted by $g_{\mbox{\sf\scriptsize b-2}}$,  maintains the equation  $g_{\mbox{\sf\scriptsize b-2}}(i+k) = 2g_{\mbox{\sf\scriptsize b-2}}(i)$.
The first line of the block headed {\sf b-2} shows the index weight, {\sf Indx}$W$, chosen here with parameter $k=5$. Starting with 1, the value doubles after each block of 5 positions. The other lines, entitled $W$, {\sf TotIndx} and {\sf IC}, are then defined as above.

A refinement of {\sf b-2}, named {\sf b-weight}, is another special weighted coding \cite{WeightedDCC21} inspired by the division by 2, but based on the function
$g_{\mbox{\sf\scriptsize b-weight}}(i)={(\sqrt[k]{2})}^{i-1}$ for $i\ge 1$, for a given parameter $k$. The function $g_{\mbox{\sf\scriptsize b-weight}}$ still maintains a fixed ratio of 2 between blocks but with a smooth hierarchy between all indices, rather than sharp differences at block boundaries.
The ratio of 2 for indices that are $k$ apart can be seen by:
$$ g_{\mbox{\sf\scriptsize b-weight}}(i+k) = {(\sqrt[k]{2})}^{i+k-1}  =  {{(\sqrt[k]{2})}^k \cdot (\sqrt[k]{2})}^{i-1}
 = 2\cdot g_{\mbox{\sf\scriptsize b-weight}}(i).$$
 
The index weight {\sf Indx}$W$ for the {\sf b-weight} method consists of real numbers rather than integers as above. The shown values correspond to $k=5$, 
so that
$g_{\mbox{\sf\scriptsize b-weight}}(i)={(\sqrt[5]{2})}^{i-1}= 1.149^{i-1}$.
This yields an average codeword length of 1.989 bits per symbol. 

The weighted approach should be applied only on  
a text that has skewed probability distributions in different portions of the file: there is a price for adjusting the model in the transition between regions of different distributions, and this overhead gets negligible only when the text becomes long enough, or if the difference between the distributions is sufficiently sharp as in this short example.

One of the features of the {\sf BWT} is that it has a tendency to reorganize the text $T$ such that {\sf BWT}$(T)$ contains several runs of repeated characters. In particular, 
for our running example, $T$, the {\sf BWT} is
$${\sf BWT}(T)={\tt t}^{7}{\tt g}{\tt t}^6{\tt a}^{14}{\tt g}^{10}{\tt t}{\tt c}^{11},$$ 
even though in the original $T$, there is not a single pair of identical adjacent characters. This tendency implies that the local distributions seem more skewed, which is advantageous to the weighted approach.

\setlength{\tabcolsep}{0.5em}

\begin{table}[!ht]
    \centering
\small
\begin{center}
\begin{tabular}
{@{\hspace{0mm}}l@{\hspace{6mm}}|c@{\hspace{6mm}}c@{\hspace{6mm}}c@{\hspace{6mm}}|c@{\hspace{6mm}}c@{\hspace{6mm}}c@{\hspace{0mm}}}
\toprule
&\multicolumn{3}{c|}{$T$}
&\multicolumn{3}{c}{${\sf BWT}(T)$}\\ [1mm]
& $k$ & header & bps & $k$ & header & bps \\ [1mm] \hline
{\sf static} & -- & 0.291 & 2.281 & -- & 0.404 & 2.394\\ 
{\sf b-adp} & -- & -- & 2.111 & -- & 0.113 & 2.224\\ 
{\sf f-adp} & -- & 0.291 & 2.111 & -- & 0.404 & 2.224\\ 
{\sf b-weight} & 5 & -- & 1.989 & 3 & 0.113 & 1.567\\
{\sf b-2} & 5 & -- & 1.981 & 3 & 0.113 & 1.562 \\
\bottomrule
\end{tabular}
\vspace{3mm}
\caption{\small\sl Storage requirements of the encoding methods
on 
${\sf BWT}(T)={\tt t}^{7}{\tt g}{\tt t}^6{\tt a}^{14}{\tt g}^{10}{\tt t}{\tt c}^{11}$.}
\label{Summerize}
\end{center}
\end{table}

We applied the same 5 methods on ${\sf BWT}(T)$.  A small amendment is necessary because the {\sf BWT} on its own is not reversible---it needs in addition a pointer to a starting point, actually the index of the last character of $T$ within {\sf BWT}$(T)$, which requires $\log n$ bits, or $\frac{\log n}{n}$ per symbol. We thus included this additional overhead of $\frac{\log 50}{50}=0.113$ bits in the header.
Note that this overhead is 
constant for all methods.
On the other hand, the static and the forward compression methods require 
information 
about the distribution probabilities,
which should be prepended to the compressed file independently from applying {\sf BWT}  or not. We use the lower bound approximation given by the {\sf IC} to 
express this information in this example.
Yet, for our experiments, we encoded the meta data information by means of a {\it universal\/} Elias \cite{Elias} $C_\delta$  code.

The comparison of the final compression results, before and after applying the {\sf BWT}, is presented in Table~2 and includes
the parameter $k$, the header size and the corresponding total storage costs including the header in bits per symbol (bps). The first three columns refer to the encoding variants on the original file $T$, and the last three columns are the results on {\sf BWT}$(T)$. 
As can be seen, while the improvement of the weighted methods with {\sf BWT} is about 21\%, the net encoding, excluding the header size, is the same for {\sf static}, {\sf b-adp} and {\sf f-adp}, regardless whether {\sf BWT} has been applied or not. The total bps results, including the header size, are identical for {\sf b-adp} and {\sf f-adp},
because the gain in the net compression by {\sf f-adp} is exactly the same as the loss incurred by the additional overhead due to the exact frequencies.   
In fact, these are not a coincidences, and in the following section we show that the compression performance is preserved for {\sf static}, {\sf b-adp} and {\sf f-adp} under {\it any\/} permutations of the symbols of $T$, hence in particular for {\sf BWT}.
Moreover, we show that the size of the compressed file including the header for {\sf b-adp} and {\sf f-adp} is the same.

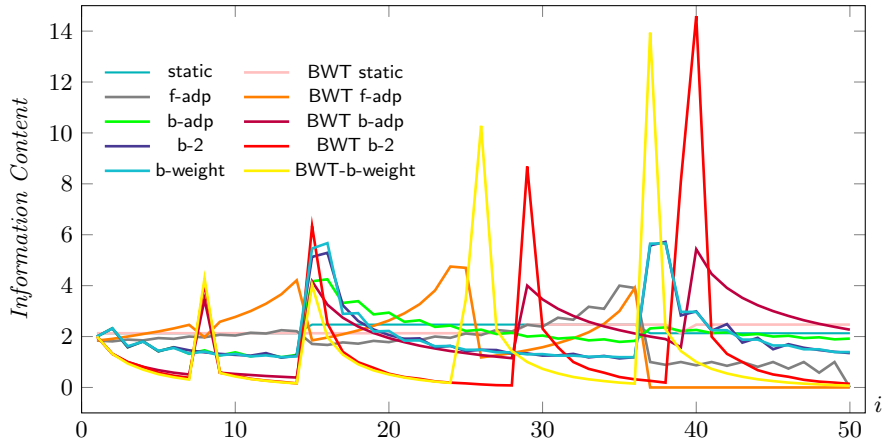
\begin{figure}[!ht]
    \centering
    \input{plotExample.tex}
    \caption{Information content per index before and after applying {\sf BWT} on $T={\tt {(at)}^7{(cg)}^{11}{(at)}^7}$.}
    \label{Example}
\end{figure}

Figure~\ref{example} is a visualization of the differences between the methods, plotting for each the information content as a function of the position $i$ for our running example. The net encoding results before and after applying {\sf BWT} are depicted in cold and warm colors, respectively. We see that the fluctuations are much more accentuated after applying {\sf BWT}.
Figure~\ref{exampleAcc} displays the cumulative values of the same data.
As can be seen from both figures, 
the backward weighted methods are more sensitive to fluctuations in the distribution, but adjust faster to changes. 
The final points of the accumulated values for the traditional methods differ only by the size of the header, and the advantage of {\sf BWT} for {\sf b-2} and {\sf b-weight} can be seen by the fact that the corresponding plots are below their counterparts without {\sf BWT}. 
The differences between {\sf b-2} and {\sf b-weight}, with and without {\sf BWT}, are so small that their plots seem to be overlapping.

\begin{figure}[!ht]
    \centering
    \input{plotExampleAcc.tex}
    \caption{Accumulated information content as a function of the size of the processed prefix, before and after applying {\sf BWT} on $T={\tt {(at)}^7{(cg)}^{11}{(at)}^7}$.}
    \label{exampleAcc}
\end{figure}
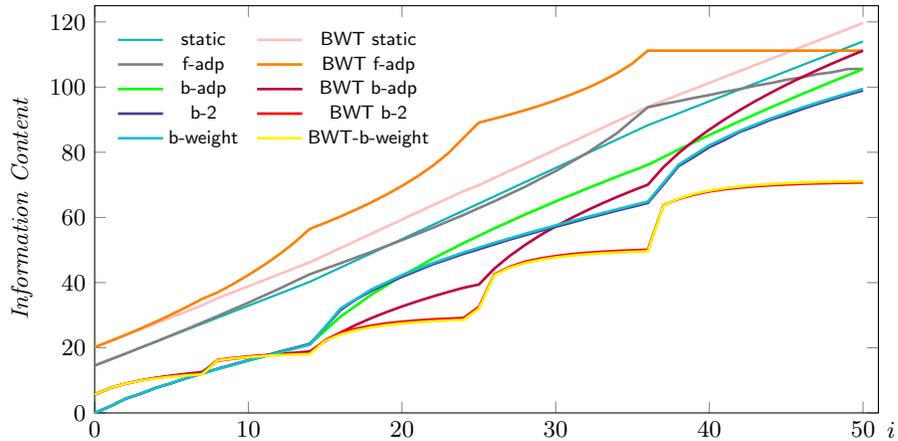






\section{Properties}

We show in this section that for the three first mentioned methods, {\sf static}, {\sf b-adp} and {\sf f-adp}, applying {\sf BWT}, or any other permutation, does not have any influence of the compression by arithmetic coding. The next section then brings empirical evidence that on the other hand, for the weighted methods, a pre-processing stage by {\sf BWT} does significantly improve the compression performance.

Arithmetic coding starts with an interval [0,1), and repeatedly narrows it as the text $T$ is being processed. The narrowing procedure is a proportional refinement of the present interval into sub-portions according to the probability distribution of the symbols of $\Sigma$. The encoding is a real number that can be selected randomly within the final interval. 
Static arithmetic coding uses a fixed probability distribution throughout the process for the interval partition, while the (backward) adaptive method updates the proportions for the corresponding partitions according to what has already been seen. 
It is well known that the static variant achieves the entropy of order 0. 
A straight\-forward corollary of this fact is that permuting $T$ does not change the {\it size\/} of the output, though the output file itself will of course be altered.
This is stated in the following lemma. 

\begin{lemma}
\label{static}
The size of the compressed file, after having applied {\sf static} arithmetic coding, is invariant under permutations of the original input.
\end{lemma}

\begin{proof}

Suppose we are encoding the text $T[1,n]$ using the {\sf static} arithmetic method. 
The notation of the weight introduced above for $\sigma=T[i]$, $W(\mathds{1},T[i],1,n)$, refers to the number of occurrences of $T[i]$ within $T[1,n]$.
For simplicity, we shall use 
${\sf occ}(\sigma)$ to denote $W(\mathds{1},\sigma,1,n)$.

Each processed letter $T[i]$ narrows the current sub-interval of $[0,1)$ by a factor equal to the probability of $T[i]$ in $P$, that is, by $\frac{1}{n}W(\mathds{1},T[i],1,n)$. The range, $r$, of the final interval after processing $T=T[1,n]$ is the product of the sizes of these intervals: 

$$r  = \prod^n_{i=1}{\frac{1}{n}W(\mathds{1},T[i],1,n)} = \frac{1}{n^n}\displaystyle\prod _{\sigma \in \Sigma}\prod_{i=1}^{{\sf occ}(\sigma)}{{\sf occ}(\sigma)}=\frac{1}{n^n}\displaystyle\prod _{\sigma \in \Sigma}{\sf occ}(\sigma)^{{\sf occ}(\sigma)},$$
where the 
middle equality
is obtained by reordering the multiplication factors by character. 
The size of the compressed file is the information content of  choosing an interval of size $r$ given by $-\log_2 r$, and it is independent of the order in which the letters appear.
\end{proof}

A similar property can be proven for traditional adaptive arithmetic coding, as follows. 

\begin{lemma}\label{adp}
The size of the compressed file, after having applied {\sf adaptive} arithmetic coding, is invariant under permutations of the original input.
\end{lemma}

\begin{proof}

To avoid the zero-frequency problem for encoding the first occurrence of a letter $\sigma$ in $T$, the number of occurrences for each $\sigma$ is initialized by 1.
For adaptive arithmetic coding, {\sf b-adp}, each processed symbol $T[i]$, $1 < i \leq n$, narrows the current sub-interval of $[0,1)$, representing the processed {\it prefix\/} $T[1,i-1]$ of $T$ of size $i-1$, by a factor equal to the probability of $T[i]$ in $T[1,i-1]$. This probability is equal to  
$$\frac{W(\mathds{1},T[i],1,i-1)+1}{m+i-1},$$ 
where $m=|\Sigma|$ taking the initial 1-values of all the characters into account. 
Multiplying the sizes of these intervals contributed by the occurrences of $T[i]$, $1 \leq i \leq n$, we get

\begin{align*}
r & = \prod^n_{i=1}{\frac{W(\mathds{1},T[i],1,i-1)+1}{m+i-1}} = \Big (\prod^n_{i=1}{\frac{1}{m+i-1}} \Big )\prod _{\sigma \in \Sigma}\prod_{i=1}^{{\sf occ}(\sigma)}{i}\\ 
& = \frac{(m-1)!}{(m+n-1)!}\prod _{\sigma \in \Sigma}{\sf occ}(\sigma)!    
\end{align*}
    
The size of the compressed file is accordingly $-\log_2 r$, so that permuting the input text will not change the size of the output.
\end{proof}

A similar proposition can be proven for {\sf f-adp}, stated as follows.

\begin{lemma}\label{f-adp}
The size of the compressed file, after having applied forward arithmetic coding, {\sf f-adp}, is invariant under permutations of the original input.
\end{lemma}

\begin{proof}

By a similar argument to Lemma~\ref{adp}, each processed letter $T[i]$, $1 < i \leq n$ narrows the current sub-interval, representing the processed {\it suffix\/} $T[i,n]$ of $T$ of size $n-i+1$, by a factor equal to the probability of $T[i]$ in $T[i,n]$, which is equal to  
$$\frac{W(\mathds{1},T[i],i,n)}{n-i+1}.$$ 
Multiplying the sizes of these intervals contributed by the occurrences of $T[i]$, $1 \leq i \leq n$,


$$r = \prod^n_{i=1}{\frac{W(\mathds{1},T[i],i,n)}{n-i+1}} = \Big (\prod^n_{i=1}{\frac{1}{n-i+1}} \Big )\prod _{\sigma \in \Sigma}\prod_{i=1}^{{\sf occ}(\sigma)}{i} = \frac{1}{n!}\prod _{\sigma \in \Sigma}{\sf occ}(\sigma)!$$
As before, the size of the compressed file is $-\log_2 r$, which is not affected by permutation of the input file.
\end{proof}

Besides proving that 
the order of the characters does not matter, only their quantity, 
we also get an exact evaluation of the difference in size between the encoded files: 
\begin{corollary}\label{forward}
The {\it forward looking} encoding, {\sf f-adp}, is better than 
the {\it backward looking} encoding, {\sf b-adp} by
$\log{m+n-1\choose{n}}$ bits.

\end{corollary}


\begin{proof}
We use $r_b$ and $r_{\!f}$ to denote the size of the final ranges of {\sf b-adp} and {\sf f-adp}, given by Lemma~3 and Lemma~4, respectively. 
The difference between the sizes of the compressed files of {\sf b-adp} and {\sf f-adp} is as follows.

\begin{align*}
    -\log r_b+\log r_{\!f} & = \log{\Big( \frac{1}{n!}\prod _{\sigma \in \Sigma}{\sf occ}(\sigma)! \Big)} -\log{\Big(\frac{(m-1)!}{(m+n-1)!}\prod _{\sigma \in \Sigma}{\sf occ}(\sigma)!\Big)}\\
     & = \log{\Big( \frac{(m+n-1)!}{n! \cdot (m-1)!}}\Big)
     = \log{m+n-1\choose{n}}
\end{align*}

\vspace*{-5mm}
\end{proof}

\vspace*{3mm}
Interestingly, the difference between {\sf b-adp} and {\sf f-adp} only depends on the size of the alphabet and the size of the file and is blind to the content itself.
In fact, on our datasets with $n=4M$ and $m=257$, including the {\sf EOF} sign, the difference between the compressed files, discarding the prelude was the constant 494 bytes, as expected. 
However, as mentioned above, the prelude for {\sf f-adp} is more costly than that for {\sf b-adp}, as it 
should include the exact frequencies of the characters. 

The number of sequences $\Big({\sf occ(\sigma_1)}, {\sf occ(\sigma_2)}, \ldots, {\sf occ(\sigma_m)}\Big)$ such that\\ 
$\sum_{i=1}^m {{\sf occ}(\sigma_i)}=n$ is equal to the number of ways to select $n$ numbers out of $n+m-1$, i.e., $n+m-1\choose{n}$.


The information content of selecting one of these choices, 
$\log {n+m-1\choose{n}}$, 
coincides therefore with the difference between {\sf f-adp} and {\sf b-adp}, which shows that basically there is no much difference between their compression efficiency. 




\section{Experimental results}

In order to study the performance of the weighted methods on {\sf BWT} transformed texts on real-life, rather than artificial data, we considered the datasets from the Pizza \& Chili corpus\footnote{\sf http://pizzachili.dcc.uchile.cl/}, a collection of  files of different nature and alphabets. 
We compared all methods before and after {\sf BWT} has been applied. 

\subsection{Compression Performance}

\begin{table}[!ht]
    \centering
\begin{center}
\begin{tabular}
{@{\hspace{0mm}}l@{\hspace{3mm}}c@{\hspace{3mm}}c@{\hspace{3mm}}|c@{\hspace{3mm}}c@{\hspace{3mm}}c@{\hspace{3mm}}c@{\hspace{3mm}}|c@{\hspace{3mm}}c@{\hspace{4mm}}c@{\hspace{3mm}}c@{\hspace{0mm}}}
\toprule
& {\sf static} &  {\sf adp} & \multicolumn{4}{c}{\sf b-2} & \multicolumn{4}{c}{\sf b-weight} \\ [1mm]
& & & $T$ & (k) & ${\sf BWT}$ & (k)& $T$ & (k) & ${\sf BWT}$ & (k) \\ [1mm]
\toprule
{\sc sources} & 69.48 & 69.48 & 64.20 & (444) & 26.44 & (24) & 64.17 & (438) & 26.30 & (22) \\			
{\sc xml} & 65.54 & 65.54 & 65.27 & (9355) & 13.11 & (24) & 65.27 & (8799) & 13.05 & (24) \\		
{\sc dna} & 24.97 & 24.98 & 24.64 & (62) & 23.02 & (34) & 24.64 & (60) & 23.01 & (34)\\		
{\sc english} & 57.07 & 57.07 & 56.55 & (4352) & 29.41 & (36) & 56.54 & (3230) & 29.33 & (36)\\			
{\sc pitches} & 69.32 & 69.32 & 56.17 & (85) & 53.34 & (119) & 56.09 & (83) & 53.26 & (120) \\			
{\sc proteins} & 52.64 & 52.65 & 51.69 & (329) & 46.71 & (39) & 51.68 & (317) & 46.62 & (38) \\			
\bottomrule
\end{tabular}
\end{center}
\vspace*{-3mm}
\caption{Compression performance on {\sf BWT} transformed files (\%) of the different methods.}
\label{Compression}
\vspace*{-3mm}
\end{table}

For our first experiments we used a 4MB prefix of each of our dataset files.
Table~\ref{Compression} shows the compression performance on the original and transformed files, defined as the size of the compressed file divided by the size of the original file, in percent. 
The column headed {\sf adp} refers to both {\sf b-adp} and {\sf f-adp}, which give identical results, as expected.
As {\sf static} and {\sf adp} are not affected by {\sf BWT} by Lemma~\ref{static} and Lemma~\ref{adp}, their results are reported only once in the first and second column, respectively. 
The following four columns refer to {\sf b-2} and the last four columns to {\sf b-weight}. The values of $k$ used to achieve the results of {\sf b-2} and {\sf b-weight} are displayed in parentheses.

As can be seen, {\sf b-2} and {\sf b-weight} outperform the traditional variants on all datasets even on the original file $T$. While they are better by up to 5\% when applied on the original file, they become better by more than 40\% in some of the cases, when applied on the {\sf BWT} transformed file. More precisely, {\sf b-2} is better than the traditional variants by 24.50\% on the average, while {\sf b-weight} is better by 24.58\% on the average. In fact, {\sf b-weight} outperforms {\sf b-2} by 0.11\% on average on all files.


\subsection{Iterative {\sf BWT}}

The {\sf BWT} 
tends to rearrange the symbols of the text so that longer runs are formed than in the original file. 
Since applying {\sf BWT} turned out to be beneficial for compressing afterwards according to methods based on the weighted paradigms, this naturally raises the question whether re-applying {\sf BWT} iteratively will arrange the symbols even better in this regard. In other words, given a text $T$, we shall try our methods on 
{\sf BWT}$^i(T)$, for $i\ge 0$, in which $i=0$ refers to  the original file $T$ itself.

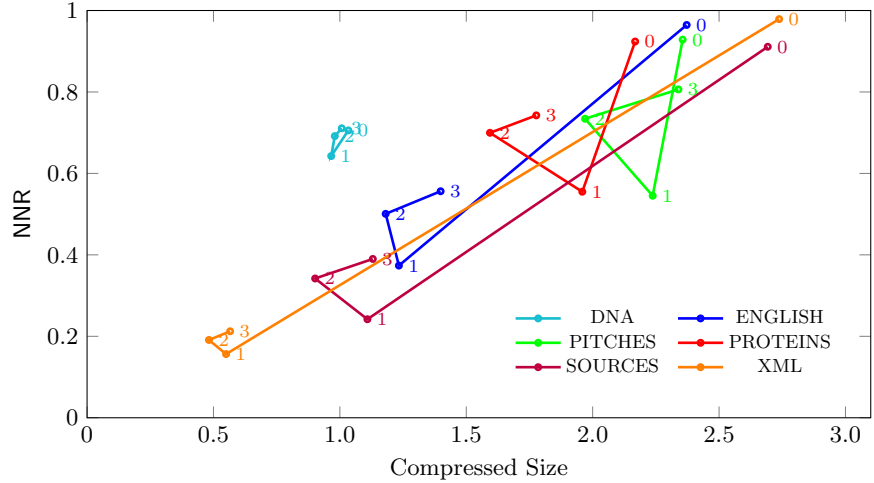
\begin{figure} [!ht]
    \centering
    \input{plotRL.tex}
    \caption{File size in million bytes and {\sf NNR}.}
    \label{RL1}
\end{figure}

To quantify this intuitive notion of {\it being more ordered\/}, we suggest the following measure. 
A {\it run\/} is a 
sequence of maximal length of the same symbol. 
For example, the string {\sf aaabaacccc} consists of 4 runs, of lengths 3, 1, 2 and 4, respectively.  
The number of runs in the text $T$, divided by its length $n$, may serve as a measure of its orderedness; we call it the normalized number of runs, {\sf NNR}, which is between 0 and 1.
The rationale is that the larger the number of runs in the input file, the shorter they are, consequently causing a longer encoding, and vice versa. 

The following graph displays, for each of the input files,  the compression performance on the $x$-axis and the {\sf NNR} on the $y$-axis of {\sf b-2} when applied on {\sf BWT}$^i(T)$, for $i=0, 1, 2, 3$. For better visibility, the points corresponding to a given input file appear in the same color and are connected and numbered by $i$, the number of times {\sf BWT} has been applied. As can be seen, the compression efficiency improves for $i=1$ and 2, but becomes worse when {\sf BWT} is applied a third time on all of the input files. The smallest {\sf NNR} is always achieved for $i=1$, and it increases for higher $i$. 




\subsection{Weighted {\sf BWT} Block Variant}

Some drawbacks of the {\sf BWT} are its streamless nature and its processing time cost.  
These weaknesses stem from the fact that the entire text $T$ should be known for the permuting process. 
To improve execution time, we suggest a block based variant and examine the block-size vs.\ compression performance trade-off. Figure \ref{Block-b2} 
displays the compression efficiency as a function of the block size for {\sf b-2}. 
The results for {\sf b-weight} produce practically the same plot. As expected, larger blocks yield better compression for textual input, whereas for special data like {\sc dna} or {\sc pitches}, the block size has only a minor influence.

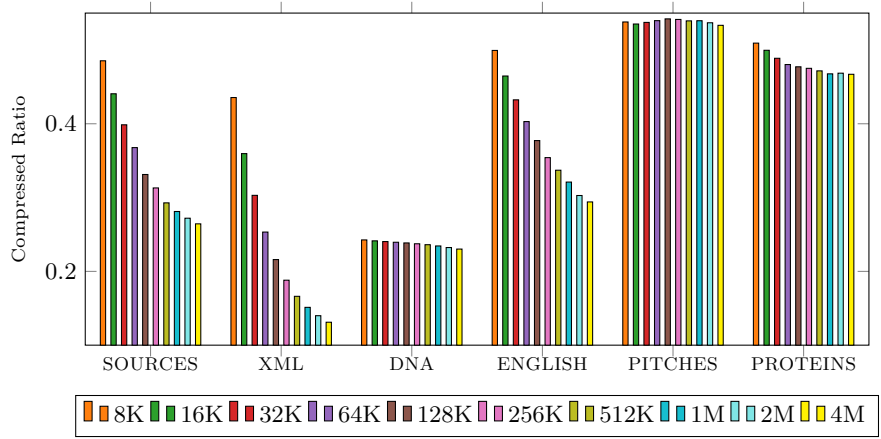
\begin{figure}[!ht]
    \centering
    \input{plotBlocksB2.tex}
    \caption{Compression efficiency of {\sc b-2} on {\sf BWT} as a function of block size.}
    \label{Block-b2}
\end{figure}



\section{Conclusion}

This paper studies the compression performance of the weighted coding on Burrows-Wheeler transformed files.
We have shown that statistical methods which treat all positions in the files evenly are indifferent to permutations in the input file, and to {\sf BWT} in particular. On the other hand, the weighted approach, being more suitable to skewed files, has been shown empirically to gain additional savings when applied to the {\sf BWT} form of the text. 

\newpage  
\bibliography{WeightedBWT}




\end{document}

%% file: plotExample.tex
\begin{center}
\pgfplotsset{compat= newest,   
width=12cm, height=7cm,
}

    \centering

    \begin{tikzpicture}
\definecolor{color1}{rgb}{1,0.498039215686275,0.0549019607843137}
\definecolor{color2}{rgb}{0.172549019607843,0.627450980392157,0.172549019607843}
\definecolor{color3}{rgb}{0.83921568627451,0.152941176470588,0.156862745098039}
\definecolor{color4}{rgb}{0.580392156862745,0.403921568627451,0.741176470588235}
\definecolor{color5}{rgb}{0.549019607843137,0.337254901960784,0.294117647058824}
\definecolor{color6}{rgb}{0.890196078431372,0.466666666666667,0.76078431372549}
\definecolor{color7}{rgb}{0.737254901960784,0.741176470588235,0.133333333333333}
\definecolor{color8}{rgb}{0.0901960784313725,0.745098039215686,0.811764705882353}
    \begin{axis}[scaled ticks=false, 
    tick label style={/pgf/number format/fixed}, 
	xtick={0, 10, ..., 50},
	ytick = {0, 2, ...,14},
	xmin=0 , xmax =51,
	ymin=-1,   ymax=15,
	every axis x label/.style={at={(current axis.right of origin)},anchor=north west},
	xlabel={$i$},
   	ylabel={\it Information Content},
    mark size=0,
    legend columns=2, 
    /tikz/column 2/.style={
                column sep=5pt,
            },
    legend style={nodes={scale=0.8, transform shape}, fill opacity=0.8, 
draw opacity=1, text opacity=1,
    at={(0.23,0.55)}, anchor=south, draw=none},
    legend style={draw=none},
]

\addplot [BlueGreen, mark size=1, mark=none, line width = 1pt, 
thick] coordinates{
(1,2.13)	(2,2.13)	(3,2.13)	(4,2.13)	(5,2.13)	(6,2.13)	(7,2.13)	(8,2.13)	(9,2.13)	(10,2.13)	(11,2.13)	(12,2.13)	(13,2.13)	(14,2.13)	(15,2.47)	(16,2.47)	(17,2.47)	(18,2.47)	(19,2.47)	(20,2.47)	(21,2.47)	(22,2.47)	(23,2.47)	(24,2.47)	(25,2.47)	(26,2.47)	(27,2.47)	(28,2.47)	(29,2.47)	(30,2.47)	(31,2.47)	(32,2.47)	(33,2.47)	(34,2.47)	(35,2.47)	(36,2.47)	(37,2.13)	(38,2.13)	(39,2.13)	(40,2.13)	(41,2.13)	(42,2.13)	(43,2.13)	(44,2.13)	(45,2.13)	(46,2.13)	(47,2.13)	(48,2.13)	(49,2.13)	(50,2.13)
};
\addlegendentry{\sf static}

\addplot [pink,  mark size=1, mark=none, line width = 1pt] coordinates{
(1,2.13)	(2,2.13)	(3,2.13)	(4,2.13)	(5,2.13)	(6,2.13)	(7,2.13)	(8,2.13)	(9,2.13)	(10,2.13)	(11,2.13)	(12,2.13)	(13,2.13)	(14,2.13)   (15,2.13)	(16,2.13)	(17,2.13)	(18,2.13)	(19,2.13)	(20,2.13)	(21,2.13)	(22,2.13)	(23,2.13)	(24,2.13)	(25,2.13)	(26,2.13)	(27,2.13)	(28,2.13)	(29,2.47)	(30,2.47)
(31,2.47)	(32,2.47)	(33,2.47)	(34,2.47)	(35,2.47)	(36,2.47)   (37,2.47)	(38,2.47)	(39,2.13)	(40,2.47)	(41,2.47)	(42,2.47)	(43,2.47)	(44,2.47)	(45,2.47)	(46,2.47)	(47,2.47)	(48,2.47)	(49,2.47)	(50,2.47)
};
\addlegendentry{\sf BWT static}

\addplot [gray,  mark size=1, mark=none, line width = 1pt] coordinates{
(1,1.84) (2,1.81) (3,1.88) (4,1.85) (5,1.94) (6,1.91) (7,2.00) (8,1.97) (9,2.07) (10,2.04) (11,2.15) (12,2.12) (13,2.25) (14,2.21) (15,1.71) (16,1.67) (17,1.77) (18,1.72) (19,1.83) (20,1.78) (21,1.91) (22,1.86) (23,2.00) (24,1.95) (25,2.12) (26,2.06) (27,2.26) (28,2.20) (29,2.46) (30,2.39) (31,2.74) (32,2.66) (33,3.17) (34,3.09) (35,4.00) (36,3.91) (37,1.00) (38,0.89) (39,1.00) (40,0.87) (41,1.00) (42,0.85) (43,1.00) (44,0.81) (45,1.00) (46,0.74) (47,1.00) (48,0.58) (49,1.00) (50,0.00)
};
\addlegendentry{\sf f-adp}

\addplot [orange,  mark size=1, mark=none, line width = 1pt] coordinates{
(1,1.84) (2,1.91) (3,2.00) (4,2.10) (5,2.20) (6,2.32) (7,2.46) (8,1.97) (9,2.58) (10,2.77) (11,3.00) (12,3.29) (13,3.66) (14,4.21) (15,1.85) (16,1.96) (17,2.09) (18,2.24) (19,2.42) (20,2.63) (21,2.91) (22,3.27) (23,3.81) (24,4.75) (25,4.70) (26,1.18) (27,1.26) (28,1.35) (29,1.46) (30,1.58) (31,1.74) (32,1.93) (33,2.17) (34,2.50) (35,3.00) (36,3.91) (37,0.00) (38,0.00) (39,0.00) (40,0.00) (41,0.00) (42,0.00) (43,0.00) (44,0.00) (45,0.00) (46,0.00) (47,0.00) (48,0.00) (49,0.00) (50,0.00)
};
\addlegendentry{\sf BWT f-adp}

\addplot [green,  mark size=1, mark=none, line width = 1pt] coordinates{
(1,2.00)	(2,2.32)	(3,1.58)	(4,1.81)	(5,1.42)	(6,1.58)	(7,1.32)	(8,1.46)	(9,1.26)	(10,1.38)	(11,1.22)	(12,1.32)	(13,1.19)	(14,1.28)	(15,4.17)	(16,4.25)	(17,3.32)	(18,3.39)	(19,2.87)	(20,2.94)	(21,2.58)	(22,2.64)	(23,2.38)	(24,2.43)	(25,2.22)	(26,2.27)	(27,2.10)	(28,2.15)	(29,2.00)	(30,2.04)	(31,1.92)	(32,1.96)	(33,1.85)	(34,1.89)	(35,1.79)	(36,1.83)	(37,2.32)	(38,2.36)	(39,2.22)	(40,2.26)	(41,2.14)	(42,2.17)	(43,2.06)	(44,2.10)	(45,2.00)	(46,2.03)	(47,1.94)	(48,1.97)	(49,1.89)	(50,1.92)
};
\addlegendentry{\sf b-adp}

\addplot [purple,  mark size=1, mark=none, line width = 1pt] coordinates{
(1,2.00)	(2,1.32)	(3,1.00)	(4,0.81)	(5,0.68)	(6,0.58)	(7,0.51)	(8,3.46)	(9,0.58)	(10,0.53)	(11,0.49)	(12,0.45)	(13,0.42)	(14,0.39)	(15,4.17)	(16,3.25)	(17,2.74)	(18,2.39)	(19,2.14)	(20,1.94)	(21,1.78)	(22,1.64)	(23,1.53)	(24,1.43)	(25,1.35)	(26,1.27)	(27,1.21)	(28,1.15)	(29,4.00)	(30,3.46)	(31,3.09)	(32,2.81)	(33,2.58)	(34,2.40)	(35,2.25)	(36,2.12)	(37,2.00)	(38,1.90)	(39,1.58)	(40,5.43)	(41,4.46)	(42,3.91)	(43,3.52)	(44,3.23)	(45,3.00)	(46,2.81)	(47,2.64)	(48,2.50)	(49,2.38)	(50,2.27)
};
\addlegendentry{\sf BWT b-adp}

\addplot [BlueViolet,  mark size=1, mark=none, line width = 1pt] 
coordinates{
(1,2.00)	(2,2.32)	(3,1.58)	(4,1.81)	(5,1.42)	(6,1.58)	(7,1.46)	(8,1.38)	(9,1.32)	(10,1.28)	(11,1.25)	(12,1.35)	(13,1.17)	(14,1.25)	(15,5.13)	(16,5.29)	(17,3.23)	(18,2.61)	(19,2.28)	(20,2.06)	(21,1.91)	(22,1.93)	(23,1.58)	(24,1.63)	(25,1.43)	(26,1.48)	(27,1.47)	(28,1.33)	(29,1.34)	(30,1.25)	(31,1.26)	(32,1.32)	(33,1.18)	(34,1.24)	(35,1.14)	(36,1.19)	(37,5.58)	(38,5.72)	(39,2.83)	(40,2.99)	(41,2.23)	(42,2.49)	(43,1.76)	(44,1.95)	(45,1.5) 
(46,1.70)	(47,1.56)	(48,1.47)	(49,1.40)	(50,1.35)};
\addlegendentry{\sf b-2}

\addplot [red,  mark size=1, mark=none, line width = 1pt] 
coordinates{
(1,2.00)	(2, 1.32)	(3,1.00)	(4,0.81)	(5,0.58)	(6,0.46)	(7,0.38)	(8,4.09)	(9,0.58)	(10,0.47)	(11,0.34)	(12,0.27)	(13,0.22)	(14,0.16)	(15,6.34)	(16,2.51)	(17,1.40)	(18,0.99)	(19,0.77)	(20,0.54)	(21,0.41)	(22,0.34)	(23,0.24)	(24,0.19)	(25,0.16)	(26,0.12)	(27,0.09)	(28,0.08)	(29,8.68)	(30,2.31)	(31,1.58)	(32,1.00)	(33,0.73)	(34,0.58)	(35,0.41)	(36,0.32)	(37,0.26)	(38,0.19)	(39,8.11)	(40,14.59)	(41,2.00)	(42,1.32)	(43,1.00)	(44,0.68)	(45,0.51)	(46,0.42)	(47,0.30)	(48,0.23)	(49,0.19)	(50,0.14)
};
\addlegendentry{\sf BWT b-2}

\addplot [color8,  mark size=1, mark=none, line width = 1pt] coordinates{
(1,2.00)	(2,2.32)	(3,1.62)	(4,1.80)	(5,1.44)	(6,1.55)	(7,1.33)	(8,1.41)	(9,1.26)	(10,1.32)	(11,1.22)	(12,1.26)	(13,1.19)	(14, 1.22)	(15,5.46)	(16,5.67)	(17,2.89)	(18,2.92)	(19,2.20)	(20,2.22)	(21,1.84)	(22,1.85)	(23,1.62)	(24,1.63)	(25,1.48)	(26,1.48)	(27,1.38)	(28,1.38)	(29,1.30)	(30,1.31)	(31,1.25)	(32,1.25)	(33,1.21)	(34,1.22)	(35,1.19)	(36,1.19)	(37,5.65)	(38,5.66)	(39,2.97)	(40,2.97)	(41,2.25)	(42,2.25)	(43,1.88)	(44,1.88)	(45,1.65)	(46,1.65)	(47,1.50)	(48,1.50)	(49,1.39)	(50,1.39)};
\addlegendentry{\small \sf b-weight}

\addplot [yellow,  mark size=1, mark=none, line width = 1pt] coordinates{
(1,2.00) (2,1.32) (3,0.94) (4,0.70) (5,0.52) (6,0.40) (7,0.31) (8,4.29) (9,0.57) (10,0.43) (11,0.33) (12,0.26) (13,0.20) (14,0.16) (15,4.02) (16,1.97) (17,1.29) (18,0.91) (19,0.67) (20,0.51) (21,0.39) (22,0.30) (23,0.23) (24,0.18) (25,3.45) (26,10.28) (27,2.27) (28,1.43) (29,1.00) (30,0.73) (31,0.55) (32,0.41) (33,0.32) (34,0.25) (35,0.19) (36,0.15) (37,13.94) (38,2.28) (39,1.43) (40,1.00) (41,0.73) (42,0.55) (43,0.42) (44,0.32) (45,0.25) (46,0.19) (47,0.15) (48,0.12) (49,0.09) (50,0.07)
};
\addlegendentry{\small \sf BWT-b-weight}


\end{axis}
\end{tikzpicture}
\end{center}

%% file: plotExampleAcc.tex
\begin{center}
\pgfplotsset{compat= newest,   
width=12cm, height=7cm,
}

    \centering

    \begin{tikzpicture}
\definecolor{color1}{rgb}{1,0.498039215686275,0.0549019607843137}
\definecolor{color2}{rgb}{0.172549019607843,0.627450980392157,0.172549019607843}
\definecolor{color3}{rgb}{0.83921568627451,0.152941176470588,0.156862745098039}
\definecolor{color4}{rgb}{0.580392156862745,0.403921568627451,0.741176470588235}
\definecolor{color5}{rgb}{0.549019607843137,0.337254901960784,0.294117647058824}
\definecolor{color6}{rgb}{0.890196078431372,0.466666666666667,0.76078431372549}
\definecolor{color7}{rgb}{0.737254901960784,0.741176470588235,0.133333333333333}
\definecolor{color8}{rgb}{0.0901960784313725,0.745098039215686,0.811764705882353}
    \begin{axis}[scaled ticks=false, 
    tick label style={/pgf/number format/fixed}, 
	xtick={0, 10, ..., 50},
	ytick = {0, 20, ...,120},
	xmin=0 , xmax =51,
	ymin=0,   ymax=125,
	every axis x label/.style={at={(current axis.right of origin)},anchor=north west},
	xlabel={$i$},
   	ylabel={\it Information Content},
    mark size=0,
    legend columns=2, 
    /tikz/column 2/.style={
                column sep=5pt,
            },
    legend style={nodes={scale=0.8, transform shape}, fill opacity=0.8, 
draw opacity=1, text opacity=1,
    at={(0.23,0.63)}, anchor=south, draw=none},
    legend style={draw=none},
]

\addplot [BlueGreen, mark size=1, mark=none, line width = 1pt, 
thick] coordinates{
(0,14.53) (1,16.37) (2,18.21) (3,20.04) (4,21.88) (5,23.72) (6,25.55) (7,27.39) (8,29.23) (9,31.06) (10,32.90) (11,34.74) (12,36.57) (13,38.41) (14,40.24) (15,42.43) (16,44.61) (17,46.80) (18,48.98) (19,51.17) (20,53.35) (21,55.54) (22,57.72) (23,59.90) (24,62.09) (25,64.27) (26,66.46) (27,68.64) (28,70.83) (29,73.01) (30,75.20) (31,77.38) (32,79.56) (33,81.75) (34,83.93) (35,86.12) (36,88.30) (37,90.14) (38,91.97) (39,93.81) (40,95.65) (41,97.48) (42,99.32) (43,101.16) (44,102.99) (45,104.83) (46,106.67) (47,108.50) (48,110.34) (49,112.18) (50,114.01)
};
\addlegendentry{\sf static}

\addplot [pink,  mark size=1, mark=none, line width = 1pt] coordinates{
(0,20.18) (1,22.01) (2,23.85) (3,25.69) (4,27.52) (5,29.36) (6,31.20) (7,33.03) (8,35.22) (9,37.05) (10,38.89) (11,40.73) (12,42.56) (13,44.40) (14,46.24) (15,48.42) (16,50.61) (17,52.79) (18,54.97) (19,57.16) (20,59.34) (21,61.53) (22,63.71) (23,65.90) (24,68.08) (25,69.92) (26,72.10) (27,74.29) (28,76.47) (29,78.65) (30,80.84) (31,83.02) (32,85.21) (33,87.39) (34,89.58) (35,91.76) (36,93.95) (37,95.78) (38,97.62) (39,99.46) (40,101.29) (41,103.13) (42,104.96) (43,106.80) (44,108.64) (45,110.47) (46,112.31) (47,114.15) (48,115.98) (49,117.82) (50,119.66)
};
\addlegendentry{\sf BWT static}

\addplot [gray,  mark size=1, mark=none, line width = 1pt] coordinates{
(0,14.53) (1,16.37) (2,18.17) (3,20.06) (4,21.91) (5,23.85) (6,25.76) (7,27.76) (8,29.72) (9,31.80) (10,33.83) (11,35.98) (12,38.10) (13,40.35) (14,42.56) (15,44.27) (16,45.94) (17,47.70) (18,49.42) (19,51.25) (20,53.04) (21,54.95) (22,56.80) (23,58.80) (24,60.75) (25,62.87) (26,64.93) (27,67.19) (28,69.39) (29,71.85) (30,74.24) (31,76.98) (32,79.64) (33,82.81) (34,85.90) (35,89.90) (36,93.81) (37,94.81) (38,95.70) (39,96.70) (40,97.57) (41,98.57) (42,99.42) (43,100.42) (44,101.23) (45,102.23) (46,102.97) (47,103.97) (48,104.55) (49,105.55) (50,105.55)
};
\addlegendentry{\sf f-adp}

\addplot [orange,  mark size=1, mark=none, line width = 1pt] coordinates{
(0,20.18) (1,22.02) (2,23.93) (3,25.93) (4,28.03) (5,30.23) (6,32.55) (7,35.01) (8,36.98) (9,39.56) (10,42.33) (11,45.33) (12,48.62) (13,52.28) (14,56.49) (15,58.34) (16,60.30) (17,62.39) (18,64.62) (19,67.04) (20,69.67) (21,72.58) (22,75.85) (23,79.66) (24,84.41) (25,89.11) (26,90.30) (27,91.56) (28,92.91) (29,94.37) (30,95.96) (31,97.70) (32,99.62) (33,101.79) (34,104.29) (35,107.29) (36,111.20) (37,111.20) (38,111.20) (39,111.20) (40,111.20) (41,111.20) (42,111.20) (43,111.20) (44,111.20) (45,111.20) (46,111.20) (47,111.20) (48,111.20) (49,111.20) (50,111.20)
};
\addlegendentry{\sf BWT f-adp}

\addplot [green,  mark size=1, mark=none, line width = 1pt] coordinates{
(0,0.00) (1,2.00) (2,4.32) (3,5.91) (4,7.71) (5,9.13) (6,10.71) (7,12.04) (8,13.50) (9,14.76) (10,16.14) (11,17.36) (12,18.68) (13,19.87) (14,21.15) (15,25.32) (16,29.57) (17,32.89) (18,36.29) (19,39.16) (20,42.10) (21,44.68) (22,47.33) (23,49.71) (24,52.14) (25,54.36) (26,56.64) (27,58.73) (28,60.88) (29,62.88) (30,64.93) (31,66.84) (32,68.80) (33,70.65) (34,72.54) (35,74.33) (36,76.15) (37,78.47) (38,80.83) (39,83.05) (40,85.31) (41,87.45) (42,89.62) (43,91.68) (44,93.78) (45,95.78) (46,97.81) (47,99.75) (48,101.72) (49,103.62) (50,105.54)
};
\addlegendentry{\sf b-adp}

\addplot [purple,  mark size=1, mark=none, line width = 1pt] coordinates{
(0,5.64) (1,7.64) (2,8.97) (3,9.97) (4,10.77) (5,11.45) (6,12.04) (7,12.55) (8,16.01) (9,16.60) (10,17.13) (11,17.61) (12,18.06) (13,18.47) (14,18.86) (15,22.03) (16,24.69) (17,27.02) (18,29.09) (19,30.96) (20,32.68) (21,34.26) (22,35.74) (23,37.11) (24,38.41) (25,39.41) (26,44.27) (27,48.17) (28,51.54) (29,54.54) (30,57.27) (31,59.77) (32,62.09) (33,64.26) (34,66.30) (35,68.23) (36,70.05) (37,75.37) (38,79.73) (39,83.54) (40,86.96) (41,90.10) (42,93.01) (43,95.73) (44,98.28) (45,100.70) (46,102.99) (47,105.17) (48,107.26) (49,109.26) (50,111.18)
};
\addlegendentry{\sf BWT b-adp}

\addplot [BlueViolet,  mark size=1, mark=none, line width = 1pt] 
coordinates{
(0,0) (1,2) 	(2,4.32) 	(3,5.9) 	(4,7.71) 	(5,9.13) 	(6,10.71) 	(7,12.17) 	(8,13.55) 	(9,14.87) 	(10,16.15) 	(11,17.4) 	(12,18.75) 	(13,19.92) 	(14,21.17) 	(15,26.3) 	(16,31.59) 	(17,34.82) 	(18,37.43) 	(19,39.71) 	(20,41.77) 	(21,43.68) 	(22,45.61) 	(23,47.19) 	(24,48.82) 	(25,50.25) 	(26,51.73) 	(27,53.2) 	(28,54.53) 	(29,55.87) 	(30,57.12) 	(31,58.38) 	(32,59.7) 	(33,60.88) 	(34,62.12) 	(35,63.26) 	(36,64.45) 	(37,70.03) 	(38,75.75) 	(39,78.58) 	(40,81.57) 	(41,83.8) 	(42,86.29) 	(43,88.05) 	(44,90) 	(45,91.5)     (46,93.2) 	(47,94.76) 	(48,96.23) 	(49,97.63) 	(50,98.98)
};
\addlegendentry{\sf b-2}

\addplot [red,  mark size=1, mark=none, line width = 1pt] 
coordinates{
(0,5.64) (1,7.64) (2,8.97) (3,9.97) (4,10.77) (5,11.36) (6,11.82) (7,12.20) (8,16.28) (9,16.87) (10,17.34) (11,17.69) (12,17.96) (13,18.18) (14,18.34) (15,22.36) (16,24.57) (17,25.85) (18,26.77) (19,27.50) (20,28.00) (21,28.39) (22,28.71) (23,28.94) (24,29.12) (25,32.50) (26,42.50) (27,44.82) (28,46.40) (29,47.40) (30,48.14) (31,48.72) (32,49.14) (33,49.46) (34,49.72) (35,49.91) (36,50.07) (37,63.65) (38,65.65) (39,66.97) (40,67.97) (41,68.65) (42,69.17) (43,69.58) (44,69.88) (45,70.11) (46,70.31) (47,70.45) (48,70.56) (49,70.65) (50,70.72)
};
\addlegendentry{\sf BWT b-2}

\addplot [color8,  mark size=1, mark=none, line width = 1pt] coordinates{
(0,0) (1,2) 	(2,4.32) 	(3,5.94) 	(4,7.74) 	(5,9.18) 	(6,10.73) 	(7,12.06) 	(8,13.47) 	(9,14.73) 	(10,16.05) 	(11,17.27) 	(12,18.53) 	(13,19.72) 	(14,20.94) 	(15,26.4) 	(16,32.07) 	(17,34.96) 	(18,37.88) 	(19,40.08) 	(20,42.3) 	(21,44.14) 	(22,45.99) 	(23,47.61) 	(24,49.24) 	(25,50.72) 	(26,52.2) 	(27,53.58) 	(28,54.96) 	(29,56.26) 	(30,57.57) 	(31,58.82) 	(32,60.07) 	(33,61.28) 	(34,62.5) 	(35,63.69) 	(36,64.88) 	(37,70.53) 	(38,76.19) 	(39,79.16) 	(40,82.13) 	(41,84.38) 	(42,86.63) 	(43,88.51) 	(44,90.39) 	(45,92.04) 	(46,93.69) 	(47,95.19) 	(48,96.69) 	(49,98.08) 	(50,99.47)
};
\addlegendentry{\small \sf b-weight}

\addplot [yellow,  mark size=1, mark=none, line width = 1pt] coordinates{
(0,5.64) (1,7.64) (2,8.97) (3,9.91) (4,10.60) (5,11.13) (6,11.53) (7,11.84) (8,16.13) (9,16.70) (10,17.13) (11,17.46) (12,17.72) (13,17.92) (14,18.08) (15,22.10) (16,24.07) (17,25.36) (18,26.28) (19,26.95) (20,27.46) (21,27.84) (22,28.14) (23,28.37) (24,28.55) (25,32.01) (26,42.28) (27,44.56) (28,45.99) (29,46.99) (30,47.72) (31,48.26) (32,48.68) (33,49.00) (34,49.24) (35,49.43) (36,49.59) (37,63.53) (38,65.81) (39,67.24) (40,68.24) (41,68.97) (42,69.52) (43,69.93) (44,70.25) (45,70.50) (46,70.69) (47,70.84) (48,70.96) (49,71.05) (50,71.12)
};
\addlegendentry{\small \sf BWT-b-weight}


\end{axis}
\end{tikzpicture}
\end{center}

%% file: plotRL.tex
\begin{center}
\pgfplotsset{compat= newest,   
width=12cm, height=7cm,
}

    \centering

    \begin{tikzpicture}
\definecolor{color1}{rgb}{1,0.498039215686275,0.0549019607843137}
\definecolor{color2}{rgb}{0.172549019607843,0.627450980392157,0.172549019607843}
\definecolor{color3}{rgb}{0.83921568627451,0.152941176470588,0.156862745098039}
\definecolor{color4}{rgb}{0.580392156862745,0.403921568627451,0.741176470588235}
\definecolor{color5}{rgb}{0.549019607843137,0.337254901960784,0.294117647058824}
\definecolor{color6}{rgb}{0.890196078431372,0.466666666666667,0.76078431372549}
\definecolor{color7}{rgb}{0.737254901960784,0.741176470588235,0.133333333333333}
\definecolor{color8}{rgb}{0.0901960784313725,0.745098039215686,0.811764705882353}
    \begin{axis}[scaled ticks=false, 
    tick label style={/pgf/number format/fixed}, 
	xtick={0, 500000, 1000000, 1500000, 2000000, 2500000, 3000000},
	xticklabels={0,0.5,1.0,1.5,2.0,2.5,3.0},
	ytick = {0, 0.2, ..., 1},
	xmin=0 , xmax =3100000,
	ymin=0,   ymax=1,
	xlabel={\small Compressed Size},
   	ylabel={\small \sf NNR},
    mark size=0,
    legend columns=2, 
    /tikz/column 2/.style={
                column sep=5pt,
            },
    legend style={nodes={scale=0.8, transform shape}, fill opacity=0.8, 
draw opacity=1, text opacity=1,
    at={(0.75,0.08)}, anchor=south, draw=none},
    legend style={draw=none},
    nodes near coords={\coordindex},
    nodes near coords align=horizontal,
    nodes near coords style={ anchor=west, font=\scriptsize    },
]

\addplot [color8,  mark size=1, mark=o, line width = 1pt] coordinates{
 (1033575,0.705610037)
 (965626,0.642575979) (979970,0.69165206) (1007484,0.710791588)
};
\addlegendentry{\small DNA}

\addplot [blue,  mark size=1, mark=o, line width = 1pt] coordinates{
  (2371893,0.964485407)[0]
  (1233486,0.373656034)[1] (1181229,0.500733376)[2] (1398551,0.556189537)[3]
};
\addlegendentry{\small ENGLISH}

\addplot [green,  mark size=1, mark=o, line width = 1pt] coordinates{
  (2355922,0.928472757)[0]
  (2237325,0.54515624)[1] (1969853,0.734207392)[2] (2339105,0.80625701)[3]
};
\addlegendentry{\small PITCHES}

\addplot [red,  mark size=1, mark=o, line width = 1pt] coordinates{
  (2167989,0.923844099)
  (1958990,0.554989576) (1593788,0.699685574) (1776879,0.742494583)
};
\addlegendentry{\small PROTEINS}

\addplot [purple,  mark size=1, mark=o, line width = 1pt] coordinates{
  (2692578,0.910964489)
  (1108953,0.241897106) (902172,0.341999292) (1130313,0.389975548)
};
\addlegendentry{\small SOURCES}

\addplot [orange,  mark size=1, mark=o, line width = 1pt] coordinates{
    (2737659,0.978893995)
    (549819,0.156415462) (481978,0.191024065) (565947,0.212334156)
};
\addlegendentry{\small XML}

\end{axis}
\end{tikzpicture}
\end{center}

%% file: plotBlocksB2.tex
\begin{center}
\pgfplotsset{compat= newest,   
width=12cm, height=6cm,
}
    \centering

    \begin{tikzpicture}
\definecolor{color1}{rgb}{1,0.498039215686275,0.0549019607843137}
\definecolor{color2}{rgb}{0.172549019607843,0.627450980392157,0.172549019607843}
\definecolor{color3}{rgb}{0.83921568627451,0.152941176470588,0.156862745098039}
\definecolor{color4}{rgb}{0.580392156862745,0.403921568627451,0.741176470588235}
\definecolor{color5}{rgb}{0.549019607843137,0.337254901960784,0.294117647058824}
\definecolor{color6}{rgb}{0.890196078431372,0.466666666666667,0.76078431372549}
\definecolor{color7}{rgb}{0.737254901960784,0.741176470588235,0.133333333333333}
\definecolor{color8}{rgb}{0.0901960784313725,0.745098039215686,0.811764705882353}
\definecolor{color9}{rgb}{0.5, 0.9, 0.9}
\begin{axis}[
    ybar,
    bar width=2pt,
    legend style={at={(0.5,-0.15)},
      anchor=north,legend columns=-1},
    ylabel={\scriptsize Compressed Ratio},
    symbolic x coords={SOURCES, XML, DNA, ENGLISH, PITCHES, PROTEINS},
    xtick=data,
    xticklabel style={font=\sc, font=\scriptsize\scshape, text height=0.5ex},
    nodes near coords align={vertical},
	ymin=0.1,   ymax=0.55,
    ]
    
\addplot [draw=black, fill=color1]
coordinates {(SOURCES,0.485353946685791) (XML,0.435561180114746) (DNA,0.242614984512329) (ENGLISH,0.49938440322876) (PITCHES,0.537954568862915) (PROTEINS,0.509284496307373)};
\addlegendentry{\small 8K}

\addplot [draw=black, fill=color2] coordinates {(SOURCES,0.440658092498779) (XML,0.359561681747437) (DNA,0.241359710693359) (ENGLISH,0.464730262756348) (PITCHES,0.535180807113647) (PROTEINS,0.49959659576416)};
\addlegendentry{\small 16K}

\addplot [draw=black, fill=color3] coordinates {(SOURCES,0.398607492446899) (XML,0.30302095413208) (DNA,0.24038290977478) (ENGLISH,0.432426929473877) (PITCHES,0.537375211715698) (PROTEINS,0.488755702972412)};
\addlegendentry{\small 32K}

\addplot [draw=black, fill=color4] coordinates {(SOURCES,0.367709398269653) (XML,0.25326132774353) (DNA,0.239452123641968) (ENGLISH,0.402979850769043) (PITCHES,0.539740562438964) (PROTEINS,0.48026180267334)};
\addlegendentry{\small 64K}

\addplot [draw=black, fill=color5] coordinates {(SOURCES,0.331275701522827) (XML,0.215987920761108) (DNA,0.23853588104248) (ENGLISH,0.377230405807495) (PITCHES,0.54215121269226) (PROTEINS,0.477210521697998)};
\addlegendentry{\small 128K}

\addplot [draw=black, fill=color6] coordinates {(SOURCES,0.313056707382202) (XML,0.188040494918823) (DNA,0.23739767074585) (ENGLISH,0.35413122177124) (PITCHES,0.541386842727661) (PROTEINS,0.475127696990967)};
\addlegendentry{\small 256K}

\addplot [draw=black, fill=color7] coordinates {(SOURCES,0.29283618927002) (XML,0.166199207305908) (DNA,0.236143112182617) (ENGLISH,0.337073802947998) (PITCHES,0.539442777633666) (PROTEINS,0.47169303894043)};
\addlegendentry{\small 512K}

\addplot [draw=black, fill=color8] coordinates {(SOURCES,0.281133413314819) (XML,0.151249408721924) (DNA,0.234441757202148) (ENGLISH,0.321006059646606) (PITCHES,0.539646625518798) (PROTEINS,0.467658758163452)};
\addlegendentry{\small 1M}

\addplot [draw=black, fill=color9] coordinates {(SOURCES,0.272042512893677) (XML,0.139930009841919) (DNA,0.232287168502808) (ENGLISH,0.302802562713623) (PITCHES,0.536862611770629) (PROTEINS,0.468557834625244)};
\addlegendentry{\small 2M}

\addplot [draw=black, fill=yellow] coordinates {(SOURCES,0.264394998550415) (XML,0.131087064743042) (DNA,0.230223178863525) (ENGLISH,0.29408597946167) (PITCHES,0.533419847488403) (PROTEINS,0.46705961227417)};
\addlegendentry{\small 4M}

\end{axis}
\end{tikzpicture}
\end{center}